\documentclass[preprint,12pt]{elsarticle}

\usepackage{amssymb}
\usepackage{graphicx}% Include figure files
\usepackage{dcolumn}% Align table columns on decimal point
\usepackage{bm}% bold math
\usepackage{amsmath,cases}
\usepackage{amssymb}
\usepackage{latexsym}
\usepackage{epsfig}
\usepackage{amsbsy}
\usepackage{array}
\usepackage{amssymb}
\usepackage{setspace}
\usepackage{bm}
\usepackage{amsthm}
\usepackage{subfigure}

\def\sint{\ifmmode{- \!\!\!\!\!\! \int}
    \else{\hbox{$- \!\!\!\! \int \ $}}\fi}

\newtheorem{theorem}{Theorem}
\newtheorem{lemma}{Lemma}

\begin{document}

\begin{frontmatter}



\title{Rational B\'{e}zier Curves Approximated by Bernstein-Jacobi Hybrid Polynomial Curves}

\author[label1]{Mao Shi}
\address[label1]{School of Mathematics and Information Science
of Shaanxi Normal University, Xi'an 710062, China}
\cortext[cor1]{Email: shimao@snnu.edu.cn}


\begin{abstract}
    A concise necessary and sufficient condition for $C^{(u,v)}$-continuity is presented. Based on weighted least squares, an analytic solution  of rational B\'{e}zier curves approximated by   polynomial curves is obtained.  An error bound between  rational B\'{e}zier curves and  B\'{e}zier ones is given. The degree reduction algorithm for  B\'{e}zier curves is used to improve approximate error.  Finally, some examples and figures are offered to demonstrate the efficiency, simplicity, and stability of the method.
\end{abstract}

\end{frontmatter}
% ----------------------------------------------------------------
\section{Introduction}
It is important to  research approximation of rational curves via polynomial ones in Computer-Aided Geometric Design (CAGD) because  differential and integral  operations of rational B\'{e}zier curves are  complicated. Since 1991, many approaches about this problem have been presented \cite{Piegl}\cite{Young}. Sederberg and Kakimoto \cite{Sederberg}  first presented hybrid  curves to achieve approximation of rational curves. Then Wang et al.\cite{Wang} analyzed its convergence condition and put forward a recursive algorithm. Wang and Zheng \cite{wangzheng} provided  methods to estiamte the error bounds of the moving control points.  Lee and Park \cite{Lee} transformed the problem into control points approximation of two B\'{e}zier curves based on the least-squares method. Using  dual constrained Bernstein  polynomial and Chebyshev polynomial, Lewanowicz etc.\cite{Lewanowicz}  derived  rational B\'{e}zier curves approximation by  polynomial curves with endpoints constraints in the $L_{2}$-norm and exploited their recursive properties. Shi and Deng \cite{Shi} introduced weighted least squares approximation to the field. Hu and Xu \cite{Hu} proposed a reparameterization-based method for polynomial curves  approximating rational B\'{e}zier curves.

If all   weights of rational B\'{e}zier curves equal one and the degree of  rational curves is more than the degree of the polynomial curves, the approximation of  rational B\'{e}zier by polynomial curves degenerate into a problem of degree reduction of B\'{e}zier curves.
Degree reduction of B\'{e}zier curves is also an important problem in CAGD. It's mainly used in the data compression and the data communication between diverse CAD/CAM systems \cite{Forrest, Farin}. Based on Jacobi polynomials, some papers \cite{Kim, Cheng, Ahn, Sunwoo, Rababah, Bhrawy} researched the  degree reduction of B\'{e}zier curves. The main idea of these papers is to convert Bernstein polynomials into Jacobi polynomials first, and then truncate the highest coefficients, or use the least squares to obtain the best approximation. Finally, convert  Jacobi polynomials back to  Bernstein forms.

In this paper, without considering the transformation from Bernstein polynomials to Jacobi polynomials, we directly apply  the Bernstein-Jacobi hybrid curves to achieve the polynomial approximation of  rational B\'{e}zier curves. In order to avoid  the integration of rational B\'{e}zier curves, we use  weighted least squares  to make the problem into an approximation of two B\'{e}zier curves. Since approximations with $C^{(u,v)}$- continuity at the end of points depend on parameterization of curves, we  find that errors between rational B\'{e}zier curves and B\'{e}zier curves sometimes can be reduced if the degree reduction of B\'{e}zier curves is used. In addition, although a necessary and sufficient condition for $C^{(u,v)}$-continuity is discussed in paper \cite{Hu}, this method invalid when $u$ or $v$ is greater than a degree of rational B\'{e}zier curve. In other words, it is only a sufficient condition. Therefore, we present a new necessary and sufficient condition for  $C^{(u,v)}$- continuity and a corresponding example given in section 4. Finally, an simple error bound is presented which independent of parameter.

The paper is structured as follows. Section 2 discusses some basic concepts and properties for developing our method. Section 3 brings a complete solution to the problem formulated  in the $L_2$ norm. Section 4 presents some numerical examples to verify the accuracy and effectiveness of the method.
\section{Preliminary}
\textbf{Definition  1.} A rational B\'{e}zier curve of degree $n$ can be presented as
 \begin{equation} \label{eq:1}
   {\mathbf{x}}({\mathbf{t}}) = \frac{{{\mathbf{p}}(t)}}
{{\omega (t)}} = \frac{{\sum\limits_{i = 0}^n {\omega _i {\mathbf{p}}_i B_i^n (t)} }}
{{\sum\limits_{i = 0}^n {\omega _i B_i^n (t)} }} = \frac{{\sum\limits_{i = 0}^n {{\mathbf{P}}_i B_i^n (t)} }}
{{\sum\limits_{i = 0}^n {\omega _i B_i^n (t)} }},
 \end{equation}
where ${\textbf{p}}_i$ are control points, $B_{i}^n(t)={n \choose i}t^i(1-t)^{n-i}$ are the Bernstein basis functions, $\omega_i$ are the corresponding positive weights  and $ \textbf{P}_i=\omega_i\textbf{p}_i$.

The integral of Bernstein basis functions is \cite{Farin}
  \begin{equation}\label{eq:4}
  \int_0^1 {B_i^n (t) = \frac{1}{{n + 1}}}, \ \ (i=0,...,n).
  \end{equation}

  If all the weights are equal and nonzero, a rational B\'{e}zier curve reduce to a B\'{e}zier curve
  \begin{equation} \label{eq:15}
 \textbf{y}(t)=\sum\limits_{i=0}^{n}\textbf{p}_iB_i^n(t).
 \end{equation}

Let ${q}(t) $ be a degree $m$ B\'{e}zier function with  control points$\{{q}_i\}_{i=0}^{m}$ and ${p}(t)$ be a degree $n$ B\'{e}zier function with coefficients   $\{{p}_i\}_{i=0}^{n}$, the product of ${q}(t) $ and  ${p}(t)$  is given by \cite{Farouki}
      \begin{eqnarray}\label{eq:6}
{p}(t)q(t)=\sum_{k=0}^{m+n}\sum\limits_{j=max(0,k-n)}^{min(m,k)}\frac{{m \choose j}{n \choose k-j}}{{m+n \choose k}}{p}_{k-j}{q}_{j}B_{k}^{m+n}(t).
%\label{eq:4}
\end{eqnarray}
For the sake of convenience, we denote the coefficients of equation (\ref{eq:6}) as following
\begin{equation} \label{eq:7}
{C}_k^{(m,q,n,p)}:=\sum\limits_{j=max(0,k-n)}^{min(m,k)}\frac{{m \choose j}{n \choose k-j}}{{m+n \choose k}}{p}_{k-j}{q}_{j}.
\end{equation}
 \textbf{Definition  2.} A degree $m$ curve
\begin{eqnarray} \label{eq:2}
{\mathbf{\tilde{ y}}}(t) = \sum\limits_{i = 0}^u {{\mathbf{q}}_i B_i^m (t)}  + t^{u+1}(1-t)^{v+1}\sum\limits_{j = 0}^{M} {{\mathbf{{\tilde{q}}}}_j J_j^{(u + 1,v + 1)} (2t - 1) + \sum\limits_{i = m - v}^m {{\mathbf{q}}_i B_i^m (t)} },\nonumber \\
\end{eqnarray}
is called a Bernstein-Jacobi hybrid curve, where $M=m-(u+v+2)$, $J_j^{(u + 1,v + 1)} (2t - 1) $ are Jacobi polynomials, ${\mathbf{q}}_i $ are constrained control points and ${\mathbf{\tilde q}}_j$ are Jacobi control points.

 The Jacobi polynomials can be represented in Bernstein form as\cite{Sunwoo}\cite{Szego}
 \begin{eqnarray*}
& J_k ^{(r + 1,s + 1)} (2t - 1)=\sum\limits_{i=0}^{k}A_iB_i^k(t), k=0,...,M,
\end{eqnarray*}
where $A_i=(-1)^{k+i}\frac{{k+r+1 \choose i}{k+s+1 \choose k-i}}{{k \choose i}}$.

By $
B_{i}^{k}\left( t \right) =\sum\limits_{j=i}^{m+i-k}\frac{{k \choose i}{m-k \choose j-i}}{{m \choose j}}{B_{j}^{m}\left( t \right)},$ we obtain
\begin{eqnarray} \label{eq:5}
&&t^{r + 1} (1 - t)^{s + 1} J_k ^{(r + 1,s + 1)} (2t - 1) \nonumber\\
 =&& \sum\limits_{i=0}^k\sum\limits_{j=i}^{M+i-k}(-1)^{k+i}\frac{{k+r+1 \choose i}{k+s+1 \choose k-i}{M-k \choose j-i}}{{m \choose r+j+1}}B_{r+j+1}^m(t),
\end{eqnarray}
which be used to convert Jacobi polynomials of degree $k$ into Bernstein polynomuals of degree $m$.

\textbf{Problem description} The polynomial approximation of rational polynomial curves based on weighted least squares  is ${\mathbf{x}}({\mathbf{t}})$ and ${\mathbf{\tilde{ y}}}(t)$  satisfying the following equation
\begin{equation}\label{eq:3}
 min \ \ \ {d} = \int_0^1 {\rho (t)\left\| {{\mathbf{x}}(t) - {\mathbf{ \tilde{y}}}(t)} \right\|^2 } dt.
\end{equation}
In particular, when the weights $\omega_i=1$ and the degree $m<n$, the problem become the degree reduction of B\'{e}zier curves.
\section{Polynomial approximation of rational curves}
In this section,  we divide the problem (\ref{eq:3}) into two parts, one for constrained control points ${\mathbf{q}}_i $ in the curve (\ref{eq:2}), the other for unconstrained control points ${\mathbf{\tilde q}}_j$.

\subsection{Constrained conditions of the approximation curve}
Although a method of solving the constrained control points based on matrix equations is provided in  paper \cite{Hu}, it can not calculate the conditions of continuity as $u> n$ and $v>n$. We give the following recursive formulas.
\begin{theorem}
  When $u+v < m-1$, given a rational B\'{e}zier curve $\emph{\textbf{x}}(t)$ and a Bernstein-Jacobi polynomial curve $\mathbf{\tilde{y}}(t)$ as in equations $(\ref{eq:1})$ and $(\ref{eq:2})$, if the equations $(\ref{eq:8})$ and $(\ref{eq:9})$ are true,
  then we call $\emph{\textbf{x}}(t)$ and $\mathbf{\tilde{y}}(t)$ satisfy $C^{(u,v)}$-continuity.
  \begin{eqnarray}\label{eq:8}
  {\mathbf{q}}_r  = && \frac{1}
{{\omega _0{m \choose r} }}\left [ {{n \choose r}\Delta ^r {\mathbf{P}}_0  - } \right.  \nonumber\\
  && \left. {\sum\limits_{i = 0}^{r - 1}  {( - 1)^{r + i} {m+n \choose r}{r \choose i}{\mathbf{C}}_i^{(m,q,n,\omega)}} - \sum\limits_{i = max(0,r-n)}^{r - 1}  {m \choose i}{n \choose r-i}\omega _{r - i} {\mathbf{q}}_i} \right ],  \nonumber\\
 &&r=0,...,u, \\
   {\mathbf{q}}_{m - s}  = && \frac{{( - 1)^s }}
{{\omega _n {m \choose s}}}\left[ {{n \choose s}\Delta ^s {\mathbf{P}}_{n - s} } \right. -    \nonumber\\
  &&  \left.{ \sum\limits_{i = 1}^s( - 1)^{s + i} {m+n \choose s}{s \choose i}{\mathbf{C}}_{m + n - s + i}^{(m,q,n,\omega)}-\sum\limits_{i = 1}^{min(s,n)} (-1)^s {m \choose s-i}{n \choose i}\omega _{n - i} {\mathbf{q}}_{m - s + i}  }   \right],   \nonumber\\ \label{eq:9}
  && s=0,...,v,
\end{eqnarray}
where $\mathbf{C}_i^{(m,q,n,\omega)} $ are given by equation (\ref{eq:7}) and ${i \choose j}:=0$ as $i<j$.
\end{theorem}
\begin{proof}
By (\ref{eq:1}) and (\ref{eq:6}), we obtain
 \begin{eqnarray} \label{eq:8b}
  \sum\limits_{i = 0}^n {\mathbf{P}}_i B_i^n (t) &=& \sum\limits_{j = 0}^m {{\mathbf{q}}_j B_j^m (t)} \sum\limits_{i = 0}^n {\omega _i B_i^n (t)}  \nonumber\\
  \  &=& \sum\limits_{k = 0}^{m + n} {{\mathbf{C}}_k^{(m,q,n,\omega)} B_k^{m + n} (t)},
\end{eqnarray}
where $\textbf{C}_k^{(m,q,n,\omega)}$ is defined by equation (\ref{eq:7}).

Differentiating both sides of  equation (\ref{eq:8b}) $r$ times with respect to $t$  and letting $t=0$ yield
\begin{equation} \label{eq:11}
{n \choose r}\Delta ^r {\mathbf{P}}_0  = {m+n \choose r}\Delta ^r {\mathbf{C}}_0^{(m,q,n,\omega)}.
\end{equation}
Here $\Delta ^r {\mathbf{C}}_0$ can be rewritten as
\begin{equation} \label{eq:12}
\Delta ^r {\mathbf{C}}_0^{(m,q,n,\omega)}  = \left( {E - I} \right)^r {\mathbf{C}}_0^{(m,q,n,\omega)}  = \sum\limits_{i = 0}^{r - 1} {( - 1)^{r + i} {r \choose i}{\mathbf{C}}_i^{(m,q,n,\omega)} }  + {\mathbf{C}}_r^{(m,q,n,\omega)}.
\end{equation}
Since $ 0 < r + s < m-1 $, we can obtain the following result by equation (\ref{eq:7})
\begin{eqnarray} \label{eq:13}
  {\mathbf{C}}_r^{(m,q,n,\omega)} &=& \sum\limits_{j = max(0,r-n)}^r {\frac{{{m \choose j}{n \choose r-j}}}
{{{m+n \choose r}}}\omega _{r - j} {\mathbf{q}}_j }  \nonumber \\
 &=& \sum\limits_{j = max(0,r-n)}^{r - 1} {\frac{{{m \choose j}{n \choose r-j}}}
{{{m+n \choose r}}}\omega _{r - j} {\mathbf{q}}_j }  + \frac{{{m \choose r}}}
{{{m+n \choose r}}}\omega _0 {\mathbf{q}}_r .
\end{eqnarray}
Then bring equations (\ref{eq:12}) and (\ref{eq:13}) into equation (\ref{eq:11}), we establish the equation (\ref{eq:8}).

Similary, we can obtain equation(\ref{eq:9}).
\end{proof}
\subsubsection{Unconstrained control points of the approximation curve}

Differentiating equation (\ref{eq:3}) with respect to $\mathbf{\tilde{q}}_k \ (k=0,..,M)$ yields the following equation:
\[
\frac{{\partial d}}
{{\partial {\mathbf{\tilde{q}}}_k }} = 2\int_0^1 {\rho (t)t^{r+1}(1-t)^{s+1}\left( {{\mathbf{x}}(t) - {\mathbf{ \tilde{y}}}(t)} \right)J_k^{( r+ 1,s + 1)} (2t - 1)dt}  = \textbf{0}.
\]
Substituting equation (\ref{eq:2}) into the above equation and letting
\[
\rho (t) =\frac{\omega (t)}{t^{r + 1}(1 - t)^{s + 1}},
\]
we obtain
\begin{eqnarray} \label{eq:14}
  &&\int_0^1 J_k^{(r + 1,s + 1)} (2t - 1){\sum\limits_{i = 0}^n {{\mathbf{P}}_i B_i^n (t)} }dt - \int_0^1 {\omega (t)} J_k^{(r + 1,s + 1)} (2t - 1)\sum\limits_{i = 0}^m {{\mathbf{Q}}_i B_i^m (t)dt} \nonumber \\
=&& \int_0^1 {t^{r + 1} (1 - t)^{s + 1} } \omega (t)J_k^{(r + 1,s + 1)} (2t - 1)\sum\limits_{j = 0}^{M} {{\mathbf{\tilde q}}_j J_j^{(r + 1,s + 1)} (2t - 1)} dt, \nonumber \\
 && k = 0,...,M
\end{eqnarray}
where
\[
{\mathbf{Q}}_i  = \left\{ {\begin{array}{lc}
   {{\mathbf{q}}_i } & {i = 0,...,r \ \ \mbox{and} \ \  i = m - s,...,m,}  \\
   0 &  \mbox{others}.\\
 \end{array} } \right.
\]
The solution of system (\ref{eq:14}) is unique, since the Bernstein polynomials are linearly independent\cite{Isaacson}.

Using equations (\ref{eq:4}), (\ref{eq:5}) and (\ref{eq:6}), the first items in the  left-hand side of equation (\ref{eq:14}) can be rewritten as
\begin{eqnarray*}
  &&\int_0^1 {\sum\limits_{i = 0}^n {{\mathbf{P}}_i B_i^n (t)}  \cdot \sum\limits_{j = 0}^k {A_j B_j^k (t)} dt}   \\
   = &&\frac{1}
{{n + k + 1}}\sum\limits_{i = 0}^{n + k} {\left( {\sum\limits_{j = \max (0,i - n)}^{\min (k,i)} {\frac{{{k \choose j}{n \choose i-j}}}
{{{k+n \choose i}}}} } A_j {\mathbf{P}}_{i - j}\right)},
 \end{eqnarray*}
and the second items are
\begin{eqnarray*}
  &&\int_0^1 {J_k^{(r + 1,s + 1)} (2t - 1)\omega (t)} \sum\limits_{i = 0}^m {{\mathbf{Q}}_i B_i^m (t)dt}  \\
   =&& \int_0^1 {\sum\limits_{j = 0}^k {A_i B_i^k (t)} \cdot} \sum\limits_{i = 0}^{m + n} {{\mathbf{ C}}_i^{(m,Q,n,\omega)} B_i^{m + n} (t)} dt  \\
   = &&\frac{1}
{{m + n + k + 1}}\sum\limits_{i = 0}^{m + n + k} {\left( {\sum\limits_{j = \max (0,i - m - n)}^{\min (k,i)} {\frac{{{k \choose j}{m+n \choose i-j}}}
{{{m+n+k \choose i}}}} A_j {\mathbf{ C}}_{i - j}^{(m,Q,n,\omega)} } \right)},  \hfill \\
\end{eqnarray*}
For the items in the left-hand side of equation (\ref{eq:14}), we have
\begin{eqnarray*}
  \int_0^1 {t^{r + 1} (1 - t)^{s + 1} \sum\limits_{i = 0}^n {\omega _i B_i^n (t)} \sum\limits_{i = 0}^k {A_i B_i^k (t)} \sum\limits_{j = 0}^{M} {{\mathbf{\tilde q}}_j J_j^{(r + 1,s + 1)} (2t - 1)} } dt \hfill \\
   =\int_0^1 {t^{r + 1} (1 - t)^{s + 1} \sum\limits_{j = 0}^{M} {{\mathbf{\tilde q}}_j J_j^{(r + 1,s + 1)} (2t - 1)} \sum\limits_{i = 0}^{n + k} {C_i^{(k,A,n,\omega)} B_i^{n + k} (t)} } dt \hfill\\
   =\sum\limits_{j = 0}^{M} {\sum\limits_{i = 0}^{n + k} {C_i^{(k,A,n,\omega)} {\mathbf{\tilde q}}_j \int_0^1 {t^{r + 1} (1 - t)^{s + 1} J_j^{(r + 1,s + 1)} (2t - 1)} } } B_i^{n + k} (t)dt \hfill\\
   = \sum\limits_{j = 0}^{M} {\sum\limits_{i = 0}^{n + k} {\sum\limits_{l = 0}^j  {\frac{{( - 1)^{j+l} {n+k \choose i}{j+r+1 \choose l}{j+s+1 \choose j-l}}}
{{(r+s+n+k+j+3){r+s+n+k+j+2 \choose r+l+i+1}}}} } }C_i^{(k,A,n,\omega)} {\mathbf{\tilde q}}_j. \ \ \ \  \hfill
\end{eqnarray*}

\subsection{Error bounds}

In this section, we  present  an error bound  for the approximation of   rational B\'{e}zier curves  by  B\'{e}zier curves.
\begin{lemma}
  If $a_k \in\mathbb{R},\  b_k>0 $ and  $c_k>0, \ 1\leq k \leq n, $ then
  \[
\frac{{\sum {a_k c_k } }}
{{\sum {b_k c_k } }} \leqslant \mathop {\max }\limits_k \frac{{a_k }}
{{b_k }}
\]
\end{lemma}
\begin{theorem}
  Let $\mathbf{x}(t)$ be a rational B\'{e}zier curve and $\mathbf{y}(t)$ be a B\'{e}zier curve, we have
\begin{equation}\label{eq:16}
\left\| {{\mathbf{x}}(t) - {\mathbf{y}}(t)} \right\| \leqslant {\mathop {\max }\limits_i \frac{{\sum\limits_{j = \max (0,i - m)}^{\min (n,i)} {{{n \choose j}{m \choose i-j}\omega _j } } \left\|{\mathbf{p}}_j  - {\mathbf{q}}_{i - j}\right\| }}
{{\sum\limits_{j = \max (0,i - m)}^{\min (n,i)} {{n \choose j}{m \choose i-j}\omega _j } }}},\
(i=0..,m+n)
\end{equation}
\end{theorem}
\begin{proof}
\begin{eqnarray*}
    &&\left\| {{\mathbf{x}}(t) - {\mathbf{y}}(t)} \right\| \hfill\\
   =&& \left\| {\frac{{\sum\limits_{i = 0}^n {\omega _i {\mathbf{p}}_i B_i^n (t)}  - \sum\limits_{i = 0}^n {\omega _i B_i^n (t)\sum\limits_{j = 0}^m {{\mathbf{q}}_j B_j^m (t)} } }}
{{\sum\limits_{i = 0}^n {\omega _i B_i^n (t)} }}} \right\| \hfill \\
    \leqslant && {\frac{{\sum\limits_{i = 0}^{m + n} {\sum\limits_{j = \max (0,i - m)}^{\min (n,i)} {\frac{{{n \choose j}{m \choose i-j}}}
{{{m+n \choose i}}}\omega _j \left\|{\mathbf{p}}_j  - {\mathbf{q}}_{i - j} \right\|} } B_i^{m + n} (t)}}
{{\sum\limits_{i = 0}^{m + n} {\sum\limits_{j = \max (0,i - m)}^{\min (n,i)} {\frac{{{n \choose j}{m \choose i-j}}}
{{{m+n \choose i}}}\omega _j B_i^{m + n} (t)} } }}} \hfill \\
\leqslant && {\mathop {\max }\limits_i \frac{{\sum\limits_{j = \max (0,i - m)}^{\min (n,i)} {{n \choose j}{m \choose i-j}} \omega _j \left\|{\mathbf{p}}_j  - {\mathbf{q}}_{i - j} \right\|}}
{{\sum\limits_{j = \max (0,i - m)}^{\min (n,i)} {{n \choose j}{m \choose i-j}} \omega _j }}},
\end{eqnarray*}
which is completed the proof.
\end{proof}

\subsection{Further discussions}
Since $C^{(u,v)}$-continuous approximation is related to  parameterized of curves, applying the degree reduction algorithm of B\'{e}zier curves sometimes can improve approximation errors. So the  algorithm can be expressed as following:

\textbf{Algorithm 1.}

\textbf{Step 1:} Compute $\{\textbf {q}\}_{i=0}^{r}$ and $\{\textbf {q}\}_{i=m-s}^{m}$  using  (\ref{eq:8}) and (\ref{eq:9}).

\textbf{Step 2:} Compute a  Jacobi-bernstein hybride curve of degree $m'$, where $m'>m$, using equation (\ref{eq:14}).

\textbf{Step 3:} Reduce  the Jacobi-bernstein hybride curve  to degree $m$ using equation ({\ref{eq:14}}).

\textbf{Step 4:} Back  Jacobi-bernstein hybride curve to  B\'{e}zier curve using  equation (\ref{eq:5}), we obtain
               \begin{equation} \label{eq:15}
                \textbf{y}(t)=\sum\limits_{i=0}^{m}\textbf{q}_iB_i^m(t),
                \end{equation}
where $\textbf{q}_i$ were given by equations (\ref{eq:5}), (\ref{eq:8}) and (\ref{eq:9}).
\section{Numerical examples}
In this section, we present several examples of approximation of rational B\'{e}zier curves by B\'{e}zier curves with constraints, which we have described in
Section $3$. The results of  example 1 can not be obtained by the method in paper \cite{Hu}. For each example, we use the Hausdoff distance and the error bound (\ref{eq:16}) for the error function, respectively.

\textbf{Example 1.} The given curve is a rational B\'{e}zier curve of degree 2 with the control points $(0,0)$, $(1.2,1.5)$, $(1,0)$ and the associated weights $1,3,1$.  We produce 8-, 9- and 10-degree B\'{e}zier curves satisfying $C^{(3,3)}$-continuity with the given curve respectively. The resulting curves are illustrated in  Fig. 1, the corresponding error distance curves are shown in Fig. 2 and the Table 1 gives  the values under the different error measures.

\begin{figure*}[!t]
   \centering
    %\begin{tabular}{cc}
    %\begin{minipage}[t]{3in}
    \includegraphics[width=2.9in]{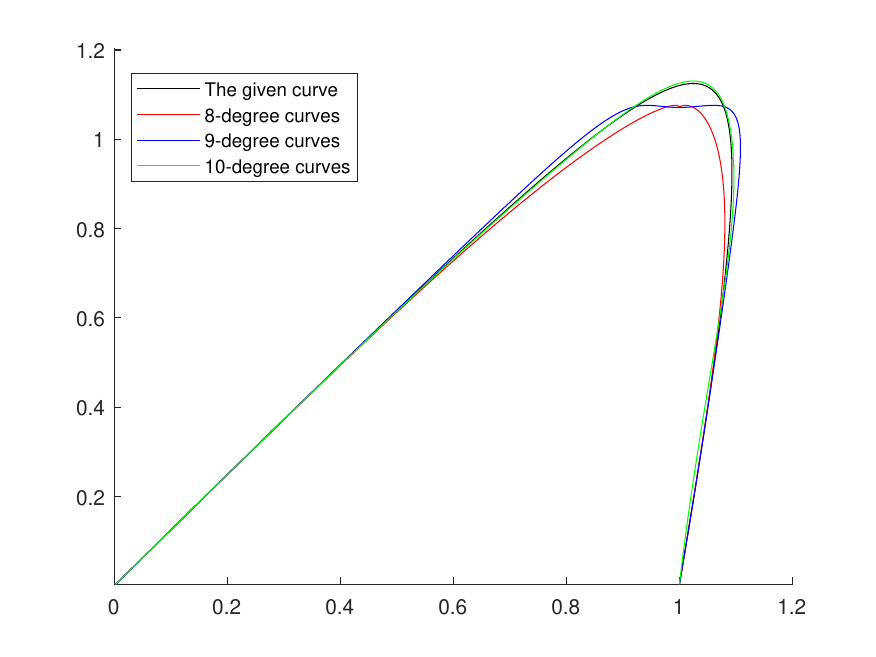}
    \caption{The rational B\'{e}zier curve of degree 2 and the resulting $C^{(3;3)}$ approximation curves of degree 8, 9 and 10. }
    %\end{minipage}
    %\begin{minipage}[t]{3.in}
    \includegraphics[width=2.9in]{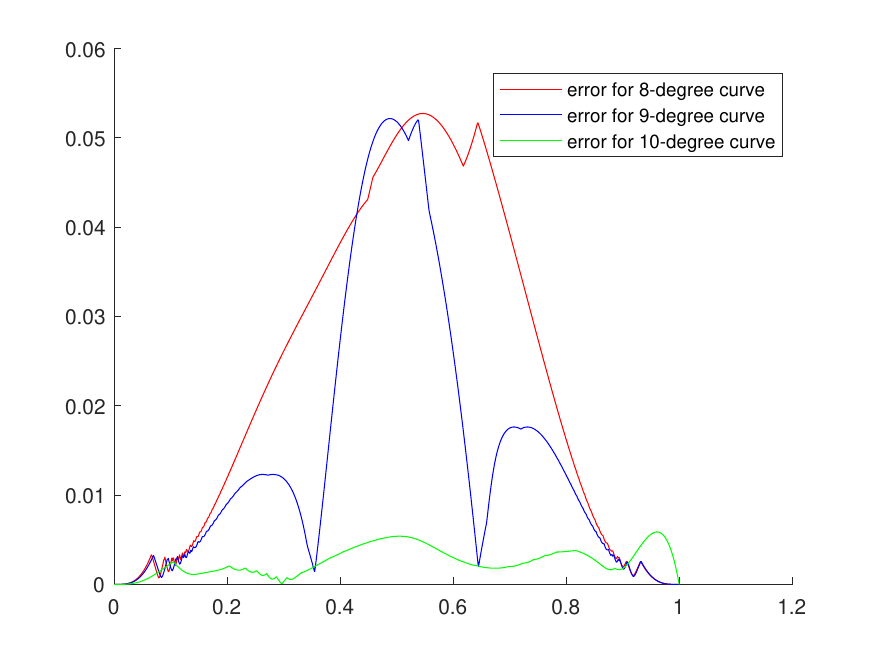}
    \caption{The corresponding Hausdoff distance  curves.}
    %\end{minipage}
%\end{tabular}
\end{figure*}

\begin{table}[th]
\caption{Error comparison of different degrees.}
{\begin{tabular}{@{}ccccc@{}} \hline
m & Errors  & Error Bounds  \\ \cline{3-5}
& under the Hausdorff distance & $L^1$ & $L^2$ & $L^{\infty}$\\ \hline
8 &  0.052717637024975 & \hphantom{0}1.597422125210526 & 1.202750122579872& 1.089151450473684 \\
9 & 0.052149525885128 & \hphantom{0}1.440000000000000 & 1.024499877989256&  0.800000000000000 \\
10& 0.005866341723445 & \hphantom{0}1.350000000000000& \hphantom{0}0.960468635614927 &0.750000000000000 \\ \hline
\end{tabular} }
\end{table}

\textbf{Example 2.} (Also Example 3 in \cite{Lewanowicz, Hu})The given curve is a rational B\'{e}zier curve of degree 9 with the control points $(17, 12)$, $(32, 34)$, $(-23, 24)$, $(33, 62)$, $(-23, 15)$, $(25, 3)$, $(30,-2)$, $(-5,-8)$, $(-5, 15)$, $(11, 8)$ and the associated weights $1, 2, 3, 6, 4, 5,3, 4, 2, 1$. We obtained a 10-degree B\'{e}zier curve satisfying $C^{(0,0)}$-continuity with the given curve, which control points are $(17,12)$, $ (40.436, 49.663)$, $(-78.372, -0.09981)$, $(160.29, 140.35)$, $(-164.8, -112.46)$, $(149.34, 188.91)$, $(-57.221, -142.63)$, $(39.424, 89.481)$, $(9.4134, -50.075)$, $(-17.15, 19.5)$, $(11,8)$.  The resulting curves are shown in  Fig. 3, the corresponding error distance curves are illustrated in Fig. 4 and the error comparison of different approximation methods  is given in Table 2.

\begin{table}[th]
\caption{Error comparison of different degrees.}
{\begin{tabular}{@{}ccccc@{}} \hline
Methods & Errors  & Error Bounds  \\ \cline{3-5}
& under the Hausdorff distance & $L^1$ & $L^2$ & $L^{\infty}$\\ \hline
Our method &  0.241932 & \hphantom{0}67.6938 & 48.9790& 38.5468 \\
 Hu and Xu's method & 0.246726 & \hphantom{0}44.7867 & 32.7662&  28.5477 \\
Lewanowicz et al.'s method & 0.317210 & \hphantom{0}54.8403& \hphantom{0}39.7418 &32.0225 \\ \hline
\end{tabular} }
\end{table}

%The approximation errors under the Hausdorff distance from our method, Hu's method, Lewanowicz's method are 0.2424, 0.246726, and 0.317210, respectively.
\begin{figure*}[!t]
   \centering
    %\begin{tabular}{cc}
    %\begin{minipage}[t]{3in}
    \includegraphics[width=2.9in]{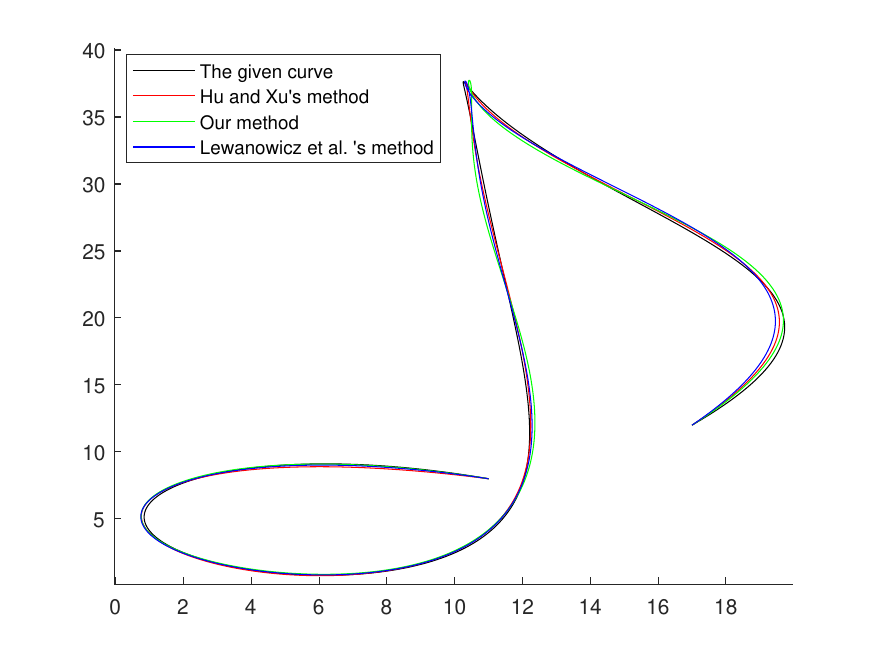}
    \caption{The rational B\'{e}zier curve of degree 9 and the resulting $C^{(0;0)}$ approximation curves of degree 10. }
    %\end{minipage}
    %\begin{minipage}[t]{3.in}
    \includegraphics[width=2.9in]{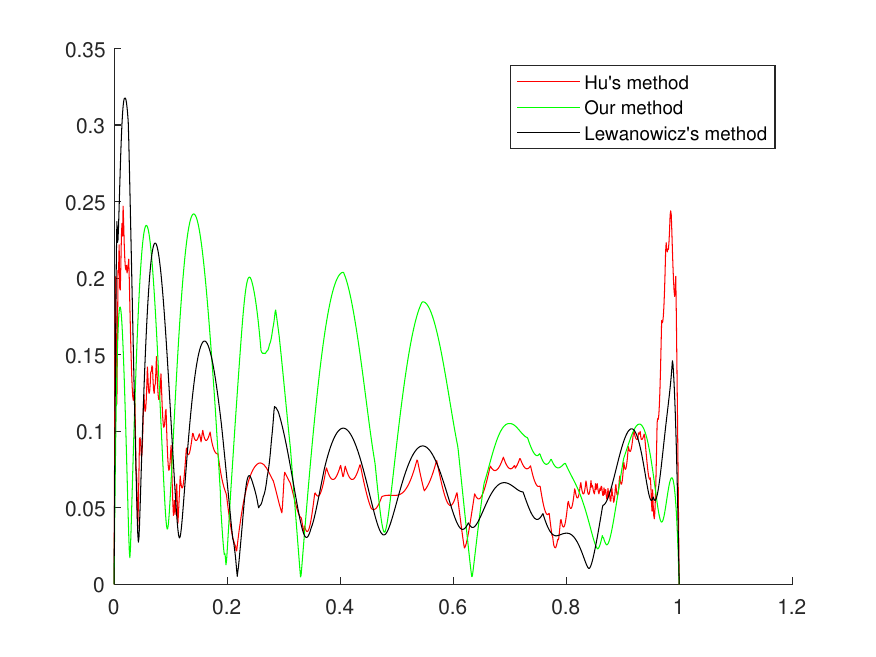}
    \caption{The error curves under the Hausdorff distance.}
    %\end{minipage}
%\end{tabular}
\end{figure*}

\textbf{Example 3.} (Also Example 3 in \cite{Hu}) The given curve is a rational b\'{e}zier curve of degree 8 with the control points $(0, 0)$, $(0, 2)$, $(2, 10)$, $(4, 6)$, $(6, 6)$, $(11, 16)$, $(8, 1)$, $(9, 1)$, $(10, 0)$ and the associated weights $1, 2, 3, 9, 12, 20, 30,4, 1$. We apply two methods to realize $C^{(0,0)}$-continuity approximation of the 5-degree B\'{e}zier curve. one is direct approximation and the other is based on degree reduction method.  It shows in Fig. 5 and Fig. 6 that the degree reduction method is better than the direct approximation method. We also find a 5-degree B\'{e}zier curve satisfying $C^{(1,1)}$-continuity with the given curve. which approximation errors under the Huasdorff distance from our method and Hu's method are $0.4982$ and $0.560612$, respectively.

 \begin{figure*}[!t]
    %\begin{tabular}{cc}
    \centering
   % \begin{minipage}[t]{3in}
    \includegraphics[width=2.9in]{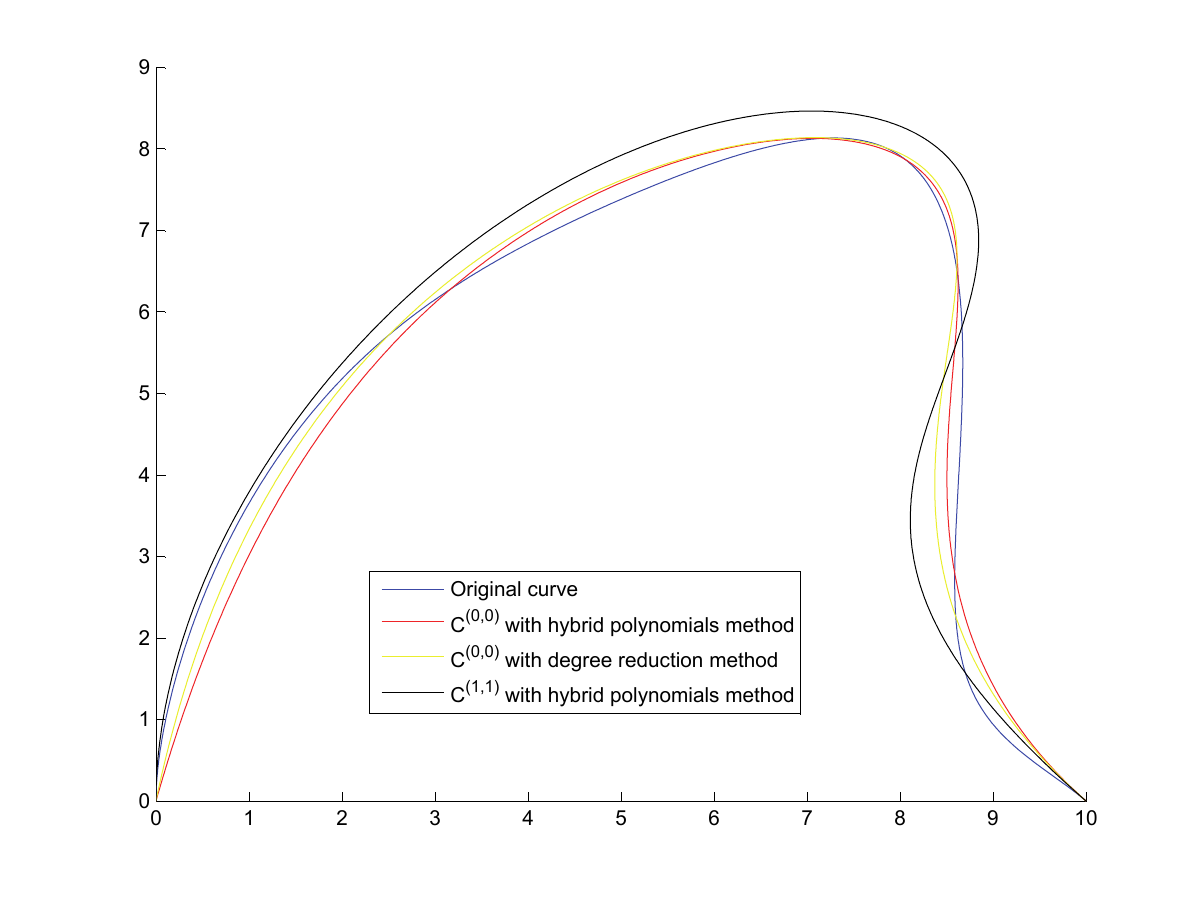}
    \caption{The rational B\'{e}zier curve of degree 8 and the resulting approximation curves of degree 5.}
    %\end{minipage}
    %\begin{minipage}[t]{3.in}
    \includegraphics[width=2.9in]{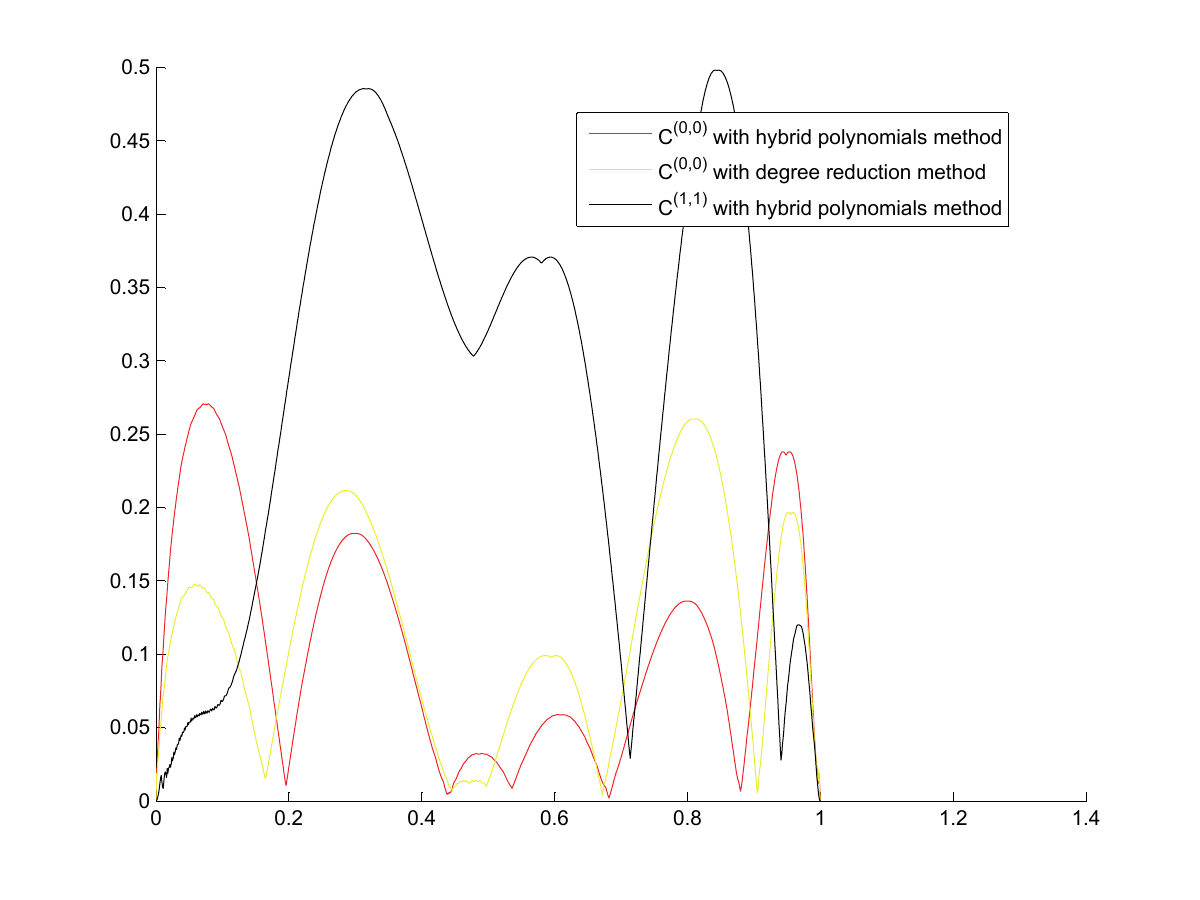}
    \caption{The error curves.  }
   % \end{minipage}
%\end{tabular}
\end{figure*}

\section*{Acknowledgements}
\label{}
The work is  supported by the Fundamental Research Funds for the Central Universities (GK201703007)

%\begin{table}[h]
%\begin{tabular}{l}
%\hline
%\textbf{Step 1:} Compute $\{\textbf {q}\}_{i=0}^{r}$ and $\{\textbf {q}\}_{i=m-s}^{m}$  using equation (\ref{eq:8}) and (\ref{eq:9}). \\
%\textbf{Step 2:} Compute a  Jacobi curve of degree $m'$ using equation (\ref{eq:14})\\
%\textbf{Step 3:} Reduce  the Jacobi curve  to degree $m$ using equation ({\ref{eq:14}})  \\
%\textbf{Step 4:} Convert control points of Jacobi curve to ones of B\'{e}zier curve using  equation (\ref{eq:5a}) \\
%\hline
%\end{tabular}
%\caption{Algorithm }
%\end{table}

% ----------------------------------------------------------------
%\bibliographystyle{amsplain}
%\bibliography{}

\end{document}